\documentclass[onecolumn,notitlepage,superscriptaddress]{revtex4-1}

\usepackage{dcolumn}
\usepackage{bm}

\usepackage{amsmath}
\usepackage{graphicx}
\usepackage{amsbsy,amsfonts}
\usepackage{amsthm}
\usepackage[colorlinks]{hyperref}
\usepackage{hypcap}
\usepackage{enumitem}

\setlist[enumerate]{itemsep=0mm}

\newcommand{\suppress}[1]{}

\def\squareforqed{\hbox{\rlap{$\sqcap$}$\sqcup$}}
\def\qed{\ifmmode\squareforqed\else{\unskip\nobreak\hfil
\penalty50\hskip1em\null\nobreak\hfil\squareforqed
\parfillskip=0pt\finalhyphendemerits=0\endgraf}\fi}

\newtheorem{theorem}{Theorem}

\newtheorem{claim}[theorem]{Claim}

\newcommand{\eq}[1]{\hyperref[eq:#1]{(\ref*{eq:#1})}}
\renewcommand{\sec}[1]{\hyperref[sec:#1]{Section~\ref*{sec:#1}}}
\newcommand{\app}[1]{\hyperref[app:#1]{Appendix~\ref*{app:#1}}}
\newcommand{\fig}[1]{\hyperref[fig:#1]{Figure~\ref*{fig:#1}}}
\newcommand{\thm}[1]{\hyperref[thm:#1]{Theorem~\ref*{thm:#1}}}
\newcommand{\lem}[1]{\hyperref[lem:#1]{Lemma~\ref*{lem:#1}}}
\newcommand{\cor}[1]{\hyperref[cor:#1]{Corollary~\ref*{cor:#1}}}
\newcommand{\defn}[1]{\hyperref[def:#1]{Definition~\ref*{def:#1}}}
\newcommand{\alg}[1]{\hyperref[alg:#1]{Algorithm~\ref*{alg:#1}}}

\usepackage{subcaption}
\DeclareCaptionType{algorithm}
\DeclareCaptionSubType{algorithm}

\usepackage[bold]{hhtensor}
\usepackage{braket}
\usepackage{subcaption}

\usepackage{algpseudocode}

    \newcommand{\inlinecomment}[1]{\Comment {\footnotesize #1} \normalsize}
    
\global\def \arxivmode {}

\ifx \arxivmode \undefined
\else
  
  \newcommand\arxivonly[1]{#1}
  \newcommand\prlonly[1]{}
\fi

\ifx \prlmode \undefined
\else
  
  \newcommand\arxivonly[1]{}
  \newcommand\prlonly[1]{#1}
  \renewcommand{\section}[1]{}
\fi

\begin{document}

\title{Using Quantum Computing to Learn Physics}
\author{Nathan Wiebe}
\affiliation{Quantum Architectures and Computation Group, Microsoft Research, Redmond, WA (USA)}
\begin{abstract}
Since its inception at the beginning of the twentieth century, quantum mechanics has challenged our conceptions of how the universe ought to work; however, the equations of quantum mechanics can be too computationally difficult to solve using existing computers for even modestly large systems.   Here I will show that quantum computers can sometimes be used to address such problems and that quantum computer science can assign formal complexities to learning facts about nature.  Hence, computer science should not only be regarded as an applied science; it is also of central importance to the foundations of science.  
\end{abstract}

\maketitle

\section{Introduction}
I first became aware that there is a deep problem at the heart of physics in a first year chemistry class.  The instructor began his very first lecture by writing an equation on the board:
\begin{equation}
i\hbar \frac{\partial \psi(r,t)}{\partial t} = \left[\frac{-\hbar^2}{2m}\nabla^2 +V(r,t) \right] \psi(r,t),\label{eq:schrodinger}
\end{equation}
and asked if anyone knew what this was.  When no one did, he continued by saying that this was the Schr\"odinger equation, the fundamental equation of quantum mechanics, and that chemisty, biology and everything around us is a consequence of this equation.  He then concluded (perhaps hyperbolically) by saying that the only problem is that no one knows how to solve it in general and for this reason we had to learn chemistry.  For me, this lesson raised an important question: how can we be certain that the Schr\"odinger equation correctly describes the dynamics of these systems if solving it is computationally intractable?  

This is the central question of this article, and in the course of addressing it we will not only challenge prevailing views of experimental quantum physics but also see that computer science is fundamental to the foundations of science.


The essence of the problem is that relatively simple quantum systems, such as chemicals or spin lattices, seem to be performing tasks that are well beyond the capabilities of any conceivable supercomputer.  This may come as a surprise since~\eq{schrodinger} is a linear partial differential equation that can be solved in time polynomial in the dimension of the vector $\psi$, which represents the ``quantum state'' of an ensemble of particles.  The problem is that the dimension of the vector $\psi$ grows exponentially with the number of interacting particles.  This means that even writing down the solution to the Schr\" odinger equation is exponentially expensive, let alone performing the operations needed to solve it.

Richard Feynman was perhaps the first person to clearly articulate that this computational intractability actually poses an opportunity for computer science~\cite{Fey82}. He suspected that these seemingly intractable problems in quantum simulation could only be addressed by using a computer that possesses quantum properties such as superposition and entanglement.  This conjecture motivated the field of quantum computing, which is the study of computers that operate according to quantum mechanical principles.

The fact that quantum computing has profound implications for computer science was most famously demonstrated by Shor's factoring algorithm, which runs in polynomial time~\cite{Sho94} whereas no polynomial time algorithm currently exists for factoring.  This raises the possibility that a fast scalable quantum computer might be able to easily break the cryptography protocols that make secure communication over the internet possible.  Subsequently, Feyman's intuition was shown to be correct by Seth Lloyd~\cite{Llo96}, who showed that quantum computers can provide exponential advantages (over the best known algorithms) for simulating broad classes of quantum systems.  Exponential advantages over existing algorithms have also been shown for certain problems related to random walks~\cite{CCD+03}, matrix inversion~\cite{HHL09,WBL12} and quantum chemistry~\cite{KW+10}, to name a few.

At present, no one has constructed a scalable quantum computer.  The main problem is that although quantum information is very flexible it is also very fragile.  Even looking at a quantum system can irreversibly damage the information that it carries.  As an example, let's imagine that you want to see a stationary electron.  In order to do so, you must hit it with a photon.  The photon must carry a substantial amount of energy (relative to the electron's mass) to measure the position precisely and so the interaction between the photon and electron will impart substantial momentum; thereby causing the electron's momentum to become uncertain.  Such effects are completely negligible for macroscopic objects, like a car, because of they have substantial mass; whereas, these effects are significant in microscopic systems.  This behavior is by no means unique to photons and electrons: the laws of quantum mechanics explicitly forbid learning information about a \emph{general} quantum state without disturbing it.

This extreme hypersensitivity means that quantum devices can lose their quantum properties on a very short time scale (typically a few milliseconds or microseconds depending on the system).  Quantum error correction strategies can correct such problems~\cite{NC10}.  The remaining challenge is that the error rates in quantum operations are still not low enough (with existing scalable architectures) for such error correction schemes to correct more error than they introduce.  Nonetheless, scalable quantum computer designs already operate close to this ``threshold'' regime where quantum error correction can lead to arbitrarily long (and arbitrarily accurate) quantum computation~\cite{Fow09}.

What provides quantum computing with its power?
Quantum computing superficially looks very much like probabilistic computing~\cite{NC10}.  Rather than having definite bit strings, a quantum computer's data is stored as a quantum superposition of different bit strings.  For example, if $\vec{v_0} = (1,0)^t$ and $\vec{v_1}=(0,1)^t$ (where $\cdot^t$ is the vector transpose) then the state $(\vec{v_0} +\vec{v_1})/\sqrt{2}$ is a valid state for a single quantum bit.  This quantum bit or ``qubit'' can be interpretted as being simultaneously 0 and 1 with equal weight.  In this case, if the state of the qubit is measured (in the computational basis) then the user will read $(0,1)^t$ with $50\%$ probability and $(1,0)^t$ with $50\%$ probability.  Quantum computers are not just limited to having equal ``superpositions'' of zero and one.  Arbitrary linear combinations are possible and the probability of measuring the quantum bitstring in the quantum computer to be $k$ (in decimal) is $|\vec{v}_k|^2$ for any unit vector $\vec{v}\in \mathbb{C}^N$.   Therefore, although quantum states are analog (in the sense that they allow arbitrary weighted combinations of bit strings) the measurement outcomes will always be discrete, similar to probabilistic computing.  

The commonalities with probabilistic computation go only so far.  Quantum states are allowed to be complex vectors (i.e. $\vec{v}\in \mathbb{C}^{2^n}$) with the restriction that $\vec{v}\vec{v}^* =1$ (here $\cdot^*$ is the conjugate--transpose operation); furthermore, any attempt to recast the quantum state as a probability density function will result (for almost all ``pure'' quantum states) in a quasi--probability distribution that has negative probability density for certain (not directly observable) outcomes~\cite{VFG+12}.  The state in a quantum computer is therefore fundamentally distinct from that of a probabilistic computer.

In essence, a quantum computer is a device that can approximate arbitrary rotations on an input quantum state vector and measure the result.  Of course giving a quantum computer the ability to perform \emph{any} rotation is unrealistic.  A quantum computer instead is required to approximate these rotations using a discrete set of gates.

The quantum gate set that is conceptually simplest consists of only two gates: the controlled--controlled not (the Toffoli gate), which can perform any reversible Boolean circuit, and an extra quantum gate known as the Hadamard gate $(H)$.  This extra gate promotes a conventional computer to a quantum computer.  The Hadamard gate which is a linear operator that has the following action on the ``0'' and ``1'' quantum bit states:
\begin{align}
H:\begin{bmatrix}1 \\ 0 \end{bmatrix} &\mapsto \begin{bmatrix}1/\sqrt{2} \\ 1/\sqrt{2} \end{bmatrix},\nonumber\\
H:\begin{bmatrix}0 \\ 1 \end{bmatrix} &\mapsto \begin{bmatrix}1/\sqrt{2} \\ -1/\sqrt{2} \end{bmatrix}.
\end{align}
In other words, this gate maps a zero--bit to an equal weight superposition of zero and one and maps a one--bit to a similar vector but with an opposite sign on one of the components.  
The Toffoli gate is a three--qubit linear operator that, at a logical level, acts non--trivially on only two computational basis vectors (i.e. bit string inputs): $110 \mapsto 111$ and $111 \mapsto 110$.  All other inputs, such as $011$, are mapped to themselves by the controlled--controlled not gate.  Using the language of quantum state vectors, the action of the Toffoli gate on the two--dimensional subspace that it acts non--trivially upon is
\begin{align}
{\rm Tof} : \begin{bmatrix}0 \\ 1 \end{bmatrix} \otimes \begin{bmatrix}0 \\ 1 \end{bmatrix} \otimes \begin{bmatrix}1 \\ 0 \end{bmatrix} \mapsto \begin{bmatrix}0 \\ 1 \end{bmatrix} \otimes \begin{bmatrix}0 \\ 1 \end{bmatrix} \otimes \begin{bmatrix}0 \\ 1 \end{bmatrix}\nonumber\\
{\rm Tof} : \begin{bmatrix}0 \\ 1 \end{bmatrix} \otimes \begin{bmatrix}0 \\ 1 \end{bmatrix} \otimes \begin{bmatrix}0 \\ 1 \end{bmatrix} \mapsto \begin{bmatrix}0 \\ 1 \end{bmatrix} \otimes \begin{bmatrix}0 \\ 1 \end{bmatrix} \otimes \begin{bmatrix}1 \\ 0 \end{bmatrix}
\end{align}

This finite set of gates is universal, meaning that any valid transformation of the quantum state vector can be implemented within arbitrarily small error using a finite sequence of these gates~\cite{Aha03}.   These gates are also reversible, meaning that any quantum algorithm that only uses the Hadamard and Toffoli gates and no measurement can be inverted.  I will make use of this invertibility later when we discuss inferring models for quantum dynamical systems.

Measurement is different from the operations described above.  The laws of quantum mechanics require that any attempt to extract information from a generic quantum state will necessarily disturb it.  Measurements in quantum computing reflect this principle by forcing the system to \emph{irreversibly} collapse to a computational basis vector (i.e. it becomes a quantum state that holds an ordinary bit string).  For example, upon measurement
\begin{equation}
\begin{bmatrix}
\sqrt{\frac{2}{3}} \\\\ - \sqrt{\frac{1}{3}}
\end{bmatrix}
{\overrightarrow{\rm measurement}}\begin{cases} [1,0]^t & \text{with $P=2/3$} \\\\ [0,1]^t & \text{with $P=1/3$}  \end{cases},
\end{equation}
where $[1,0]^t$ represents logical $0$ and $[0,1]^t$ is the logical $1$ state.
There clearly is no way to invert this procedure.  In this sense, the act of measurement is just like flipping a coin.  Until you look at the coin, you can assign a prior probability distribution to it either being heads or tails, and upon measurement the distribution similarly ``collapses'' to either a  heads or tails result. 
 Quantum computation thus combines elements of reversible computing with probabilistic computing wherein the inclusion of the Hadamard gate introduces both the ability to have non--positive quantum state vectors and superposition (the ability for quantum bits to be in the state 0 and 1 simultaneously) and thereby promotes the system to a universal quantum computer. 

Although quantum computing is distinct from probabilistic computing, it is unclear whether quantum computation is fundamentally more powerful.  Using the language of complexity theory, it is known that the class of decision problems that a quantum computer can solve efficiently with success probability greater than $2/3$, \textsf{BQP}, obeys $\textsf{BPP} \subseteq \textsf{BQP} \subseteq \textsf{PSPACE}$.  Here \textsf{BPP} is the class of decision problems that can be efficiently solved with success probability greater than $2/3$ using a probabilistic Turing machine, and \textsf{PSPACE} is the class of problems that can be solved using polynomial space  and (possibly)  exponential time using a deterministic Turing machine.  It is \emph{strongly} believed that $\textsf{BPP} \subset \textsf{BQP}$ since exponential separations exist between the number of times that quantum algorithms and the best possible conventional algorithms have to query an oracle to solve certain problems~\cite{NC10,CCD+03}; however, the precise relationship between the complexity classes remains unknown.


The apparent difficulty of simulating large quantum systems creates an interesting dilemma for quantum computing: although true quantum computers are outside of our capabilities at present, it would appear that purpose built analog devices could be used to solve problems that are truly difficult for conventional computers.  Recent experiments by the NIST group have demonstrated that a two--dimensional lattice of ions with over two hundred and seventeen ions with programmable interactions can be created~\cite{BSK+12}.  This device can be thought of as a sort of \emph{analog quantum simulator} in the sense that these interactions can be chosen such that the system of ions approximate the dynamics of certain condensed physics models known as Ising models.  

On the surface, it would seem that this system of ions may be performing calculations that are beyond the reach of any conceivable supercomputer.
If the quantum state for such a system were ``pure'' (which roughly speaking means that it is an entirely quantum mixture) then the state vector would be  in $\mathbb{C}^{2^{217}}$.  If the quantum state vector were expressed as an array of single precision floating point numbers then the resultant array would occupy roughly $1.3\times 10^{61}$ megabytes of memory. It is inconceivable that a conventional computer could even store this vector, let alone solve the corresponding Schr\"odinger equation.   

A conventional computer does provide us with something that the analog simulator does not: the knowledge that we can trust the computer's results.  The analog simulator has virtually no guarantees attached that the dynamical model that physicists believe describes the system is actually correct to a fixed number of digits of accuracy.  Finding a way to use computers to inexpensively certify that such quantum devices function properly not only remains an important challenge facing the development of new quantum technologies but it also poses a fundamental restriction on our abilities to understand the dynamics of large quantum systems, provided that Feynman's conjecture is correct.

\section{Why is Certification A Problem?}
In essence, the central question of this paper is that of whether we can practically test quantum mechanical models for large quantum systems.  Although this is a question about physics, it can only be resolved by using the tools computer science.  In particular, we will see that formal costs can be assigned to model inference and that computational complexity theory can provide profound insights about the limitations of our ability to model nature using computers.  
%
%

The importance of computational complexity to the foundations of physics has only recently been observed. In particular, the apparent emergence of thermodynamics in closed quantum systems is, in some cases, a consequence of the cost of distinguishing the ``true'' quantum mechanical probability distribution from the thermodynamical predictions scaling exponentially with the number of particles in the ensemble~\cite{UWE13,SF12}.  This is especially intriguing since it provides further evidence for the long held conjecture that there is an intimate link between the laws of thermodynamics in physics and computational complexity~\cite{Mer02}.

The problem of modeling, or certifying, quantum systems is essentially the problem of distinguishing distributions.
Although the problem of distinguishing distributions has long been solved, the problem of doing so under particular operational restrictions can actually be surprisingly subtle (even in the absence of quantum effects).  To see this, consider the following problem of deciding whether a device yields samples from a probability distribution $p(x)$ or from another distribution $q(x)$.
The probability of correctly distinguishing which one of two possible discrete probability distributions $p(x)$ and $q(x)$ based on a single sample, using the best possible data processing technique, is given by the variational distance:
\begin{equation}
P_{\rm dist} = \frac{1}{2} \left(1+  \frac{1}{2}\sum_{x} |p(x)-q(x)|\right).\label{eq:vdist}
\end{equation}
Given any non--zero bias in this measurement, the Chernoff bound shows that the success probability can be boosted to nearly $100\%$ by repeating the experiment a logarithmically large number of times.  

For typical quantum systems, $P_{\rm dist}$ is of order one~\cite{UWE13} which means that in principle such systems are easy to distinguish.
In practice, the processing method that efficiently distinguishes $p(x)$ from $q(x)$  may be impractically difficult to find in quantum mechanical problems.  
Thus distinguishing typical quantum systems can still be challenging.

As a motivating example, assume that you have a dice with $10,000$ sides that is promised to either be fair (i.e. $p(x)=1/10,000$) or that the probability distribution describing the dice was itself randomly drawn from another distribution such that $\mathbb{E}_x(p(x))=1/10,000$ and the variance in the probabilities obeys $\mathbb{V}_x(p(x)) = 10^{-8}$.  This problem is substantially different from the base problem of distinguishing two known distributions because $q(x)$ is unknown but information about the distribution that $q(x)$ is drawn from is known.  A number of samples that scales at least as the fourth--root of the dimension are needed to distinguish the distributions with high probability (over samples and $p(x)$)~\cite{UWE13,GKA+13}.  This problem is not necessarily hard: a rough lower bound on the number of samples needed is $10$.

Distinguishing typical probability distributions that arise from quantum mechanics from the uniform distribution is similar to the dice problem except that in large quantum systems ``the dice'' may have $2^{200}$ sides or more.  In the aforementioned example a minimum of roughly ${2^{50}}$ samples will be needed to distinguish a quantum distribution from the uniform distribution with probability greater than $2/3$~\cite{UWE13,GKA+13}.  Although collecting $2^{50}$ samples may already be prohibitively large, the number of samples needed at least doubles for every $4$ qubits that are added to the quantum system.  This means that distinguishing typical quantum probability becomes prohibitively expensive as quantum bits (i.e. particles) are added to the system, despite the fact that $P_{\rm dist}$ is of order one.

In short, quantum dynamics tend to rapidly scramble an initial quantum state vector, causing the predictions of different quantum mechanical models to become practically indistinguishable given limited computational resources.  We will see below that complexity theory provides strong justification for the impracticality of distinguishing quantum distributions in general.
In the following I will assume that the reader is familiar with standard complexity classes $\textsf{P}$, $\textsf{NP}$, the polynomial hierarchy, $\textsf{\#P}$ and so forth.  For brevity, the term efficient will mean ``in polynomial time''.

\section{Boson Sampling}
Scott Aaronson and Alex Arkhipov provided the strongest evidence yet that quantum systems exist that are exponentially difficult to simulate using conventional computers in their seminal 2011 paper on ``Boson Sampling''~\cite{AA11}.  Their work proposes a simple experiment that directly challenges the extended Church--Turing thesis, which conjectures that any physically realistic model of computation can be efficiently simulated using a probabilistic Turing machine.  Their Boson sampler device is not universal and is trivial to simulate using a quantum computer.  The remarkable feature of this device is that it draws samples from a distribution that a conventional computer cannot efficiently draw samples from under reasonable complexity theoretic assumptions.  In particular, an efficient algorithm to sample from an arbitrary Boson sampler's outcome distribution would cause the polynomial hierarchy to collapse to the third level, which is widely conjectured to be impossible.

A Boson sampling experiment involves preparing $n$ photons (which are Bosons) in $n$ different inputs of a network of beam splitters, which are pieces of glass that partially reflect and transmit each incident photon.  An illustration of this device is given in \fig{boson} (a).  The Boson sampler is analogous to the Galton board, which is a device that can be used to approximately sample from the binomial distribution.  The main difference between them is that the negativity of the quantum state (manifested as interference between the photons) causes the resultant distribution to be much harder to compute than the binomial distribution, as we will discuss shortly.  At the end of the protocol, detectors placed in each of $m$ different possible output ``modes'' count the number of photons that exit in each mode.

The sampling problem that needs to be solved in order to simulate the action of a Boson sampler is as follows.  Let $A\in \mathbb{C}^{m,n}$ be the transition matrix mapping the input state to the output state for the Boson sampling device.  Let us define $S=(s_1,\ldots,s_m)$ to be the experimental outcome observed, i.e. a list of the number of photons observed in mode $1,\ldots,m$.  Then let $A_S\in \mathbb{C}^{n\times n}$ be the matrix that has $s_1$ copies of the first row of $A$, $s_2$ copies of the second and so forth.  Then the probability distribution over $S$ given $A$ is~\cite{AA11}
\begin{equation}
\Pr(S|A)= \frac{|\rm{Per}(A_S)|^2}{s_1!\cdots s_m!},\label{eq:sampprob}
\end{equation}
where ${\rm Per}(\cdot)$ is the permanent of a matrix, which is like a determinant with the exception that positive cofactors are always used in the expansion by minors step.  It is important to note, however, that a quantum computer is not known to be able to efficiently learn these probabilities and in turn the permanent.  

Permanent approximation is known to be $\textsf{\#P}$  complete~\cite{Val79,AA11}, where $\textsf{\#P}$ is the class of problems associated with counting the number of satisfying assignments to problems whose solutions can be verified efficiently. Such problems can be extremely challenging (much harder than factoring in the worst cases), which strongly suggests that the distribution for these experiments will be hard to find for certain cases if $n\gg 100$.  The work of Aaronson and Arkhipov use the hardness of permanent approximation (albeit indirectly) to show that even drawing a sample from this distribution can be computationally hard under reasonable complexity theoretic assumptions.  In contrast, drawing a sample from the quantum device is easy (modulo the engineering challenges involved in preparing and detecting single photons~\cite{BFR+13}) since it involves only measuring the number of photons that leave the Boson sampling device in each of the $m$ possible modes.
This provides a compelling reason to believe that there are simple physical processes that are too difficult to simulate using any conceivable conventional computer.

Although the Boson sampling distribution can be hard to compute,  it can be easily to distinguished from the uniform distribution because information is known about the asymptotic form of the distribution is known~\cite{AA13}.  Despite this, there are other distributions that are easy to sample from but are hard to distinguish from typical Boson sampling distributions~\cite{AA13}.  Thus certifying Boson samplers can still be a difficult.  A more general paradigm is needed to approach certifying, or more generally learning, an underlying dynamical model.  Bayesian inference provides a near ideal framework to solve such problems.
\begin{figure*}[t!]
\includegraphics[width=\textwidth]{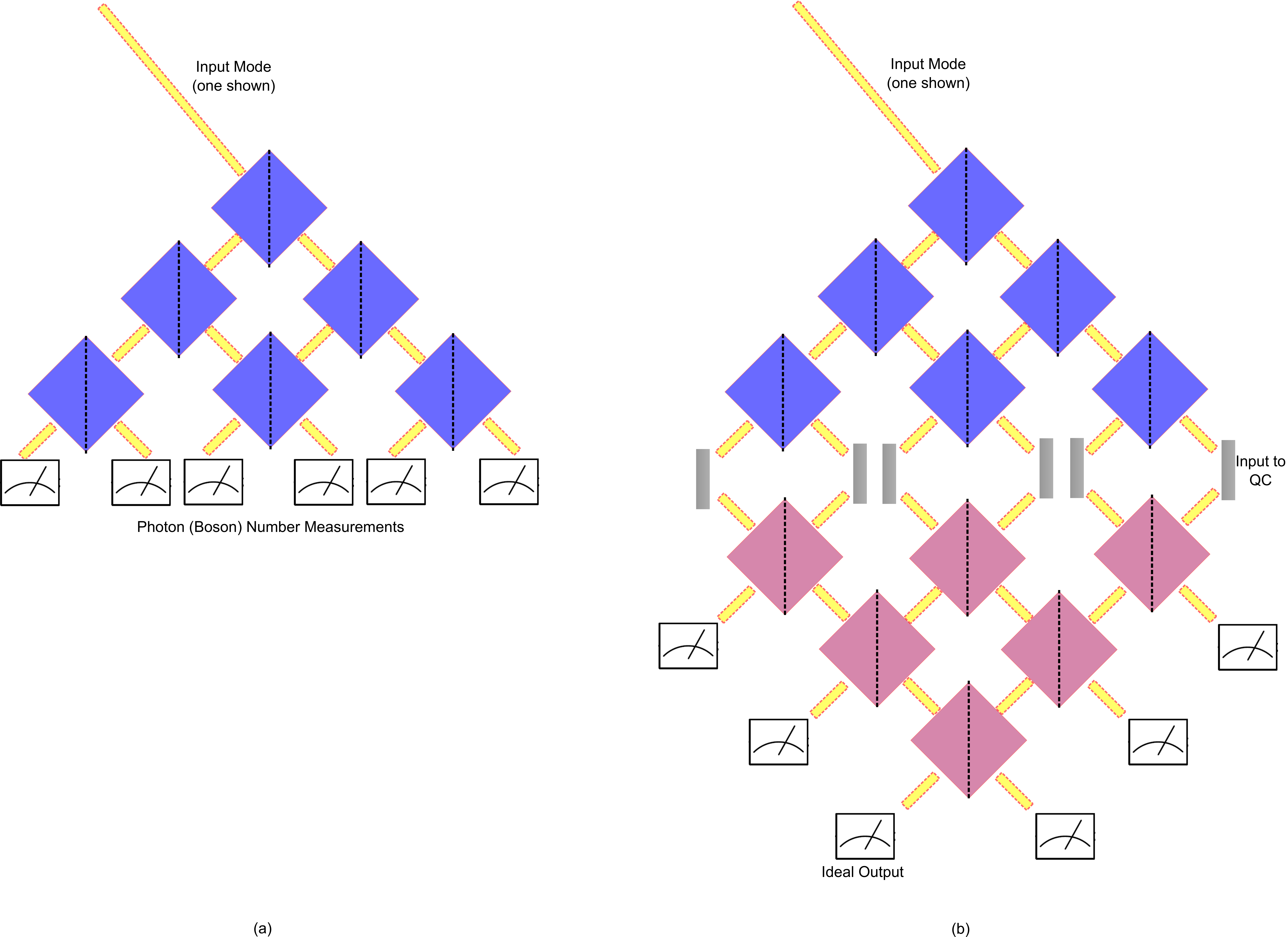}
\caption{(a) Schematic representation of boson sampling experiment where blue squares are beam splitters.  (b) Quantum inference experiment where pink squares represent a quantum computer simulation of the inverse of a hypothetical transformation carried out by (a).  If the guess is correct (or nearly so) all the output bosons will be measured in the same mode that they were input with high probability.}\label{fig:boson}
\end{figure*}

\section{Bayesian Inference}
Bayes' theorem provides a convenient way to find the probability that a particular model for a system is true given an observed datum $D$ and a set of prior assumptions about the model $H$.  The prior assumptions about the model can be encoded as a probability distribution $\Pr(H)$ that can be interpreted as a probability that you subjectively assign to represent your belief that the true model is $H$.  Constraints on the classes of models allowed can also be naturally included in this framework by changing the support of the prior distribution (i.e. disallowed models are given zero a priori probability).  The probability that model $H$ is correct, given the prior distribution and $D$ is

\begin{equation}
\Pr(H|D) =\frac{\Pr(D|H)\Pr(H)}{\Pr(D)},\label{eq:bayes}
\end{equation}
where $\Pr(D)$ is a normalizing constant.  The resultant distribution is known as the posterior distribution, which then becomes the prior distribution for the next experiment.  This procedure is called ``updating'' and $\Pr(D|H)$ is called the likelihood function.
  
There is a strong connection in Bayesian inference between learning and simulation.  This link is made explicit through the likelihood function, $\Pr(D|H)$, whose computation can be thought of as a simulation of $H$.  It further suggests that if the likelihood function cannot be computed efficiently then updating your beliefs about the correct model for the system, given an observation of the system,  is also not efficient. This can make Bayesian inference intractable.

Let's return to the problem of testing models for quantum systems.  In this context, Bayesian inference can be extremely difficult to perform using conventional computers because the likelihood function is evaluated by simulating a quantum system, which we have argued can be exponentially expensive modulo reasonable complexity theoretic assumptions.  This raises an interesting philosophical question: does computational complexity place fundamental limitations on our ability to understand systems in nature?
Although this may be difficult to show in general (owing to the plethora of ways that a quantum system can be probed and the data can be analyzed), it is possible to show that this conjecture holds in certain limiting cases:

\begin{claim}
There exist finite-dimensional models $H_a$ and $H_b$ for a quantum mechanical system and fixed observables $\mathcal{M}$ such that models $H_a$ and $H_b$ cannot be efficiently distinguished using Bayesian inference based on outcomes of $\mathcal{M}$ without causing the polynomial hierarchy to collapse and yet quantum computing can be used to efficiently distinguish these models using  Bayesian inference and the same $\mathcal{M}$.
\end{claim}
The proof of the first half of this claim is straight forward.  Let us consider $H_a$ and $H_b$ to be Boson sampling experiments with different matrices $A_a$ and $A_b$ which are taken to be Gaussian random matrices.  Let us further take $\mathcal{M}$ to represent the measurement of the number of Bosons in each of the modes.  We saw that the likelihood of a particular outcome $D$ being observed involves calculating the permanents of $(A_a)_S$ and $(A_b)_S$~\cite{AA11}.  Efficient approximation of the permanent causes the polynomial hierarchy to collapse~\cite{AA11} and hence, modulo reasonable complexity theoretic assumptions, Bayesian inference cannot be used to distinguish between these models efficiently using $\mathcal{M}$.

A quantum computer cannot be expected to inexpensively compute the likelihood function directly because that entails computing a permanent: a task for which neither a fast quantum or conventional algorithm is known.  So how is it that a quantum computer can provide an advantage?  The answer is that a quantum computer can be used to change the problem to one that is easily decidable.  The key freedom that a quantum computer provides is the ability to perform basis transformations that, in effect, allow $\mathcal{M}$ to be transformed into a different measurement operator whose outcomes provide clear evidence for some models over others.    The presence of a quantum computer (that is coupled to the experimental system) therefore gives us the ability to solve the problem of model distinction in a completely different way.

Boson samplers can be efficiently simulated by a quantum computer.  Indeed, a Boson sampler is simply a linear optical quantum computer without the ability to adaptively control quantum gates based on prior measurements (which is needed for universality).  This means that a universal quantum computer trivially has all of the computational power of a Boson sampler~\cite{AA11}.  Furthermore, closed system quantum dynamics are always be invertable (unless the system is measured) as we discussed in the context of quantum computing above.  This means that we can scramble the initial state $\vec{v}$ by running it through the Boson sampler and attempt to approximately unscramble it using a quantum computer that simulates the inverse of the map that the Boson sampler $H_a$ would perform.  This inversion experiment is illustrated in \fig{boson} (b).  The forward evolution under the experimental system and approximate inversion under the quantum computer of the initial quantum state vector is
\begin{equation}
\vec{v} \mapsto U_a^{-1} U \vec{v},\label{eq:6}
\end{equation}
where $U_a$ is the quantum transformation generated by $H_a$, $U_b$ is the quantum transformation generated by $H_b$ and $U$ is an unknown transformation promised to be either $U_a$ or $U_b$.  If, for example, $U=U_a$ then this protocol performs:
\begin{equation}
\vec{v} \mapsto U_a^{-1} U_a \vec{v} = \vec{v},\label{eq:7}
\end{equation}
meaning that we will always observe as output the same state that was used as the input.  In contrast, if $U=U_b$ then this process leads to
\begin{equation}
\vec{v} \mapsto U_a^{-1} U_b \vec{v} \ne \vec{v}.\label{eq:8}
\end{equation}
In fact, with high probability over the matrices $A_a$ and $A_b$ this transformation will lead to a vector that is nearly orthogonal to $\vec{v}$ if $N$ is large~\cite{AA13,GKA+13}, i.e. $\vec{v}^* U_a^{-1} U_{b}\vec{v} \approx 0$, where $*$ is the conjugate transpose operation.  In Boson sampling experiments, $\vec{v}$ is taken to be a Boson number eigenstate meaning that a measurement of the number of Bosons in each mode will yield the same value every time.  This means that we can easily use the same $\mathcal{M}$  to determine whether $U_a^{-1}U$ maps $\vec{v}$ back to itself, which should happen with probability $1$ if $U=U_a$ and with low probability if $U=U_b$.  Thus the decision problem can be solved efficiently using Bayesian inference in concert with quantum computing, which justifies the above claim. 

Quantum computation can therefore be used to convert seemingly intractable problems in modeling quantum systems into easily decidable ones.  Also, by using a quantum computer in place of an experimental device we can assign a computational complexity to performing experiments.  This allows us to characterize facts about nature as ``easy'' or ``hard'' to learn based on how computationally expensive it would be to do so using an experimental system coupled to a quantum computer; furthermore such insights are valuable \emph{even in absentia of a scalable quantum computer}.

\section{A New Way to Approach Physics}
A quantum computer is more than just a computational device: it also is a universal toolbox that can emulate any other experimental system permitted by quantum theory.  Put simply, if a quantum computer were to be constructed that accepts input from external physical systems then experimental physics would become computer science.  When seen in this light, computational complexity becomes vital to the foundations of physics because it gives insights into the limitations of our ability to model, and in turn understand, physical systems that are no less profound than those yielded by the laws of thermodynamics.  

It is assumed that the reader is familiar with quantum computing in the following.  In particular,~\thm{likelihood} and~\alg{qhl} are provided for the benefit of experts in quantum computation.  The details of these results can, nonetheless, be safely ignored by the reader.

This approach to learning physics using quantum computers is simply an elaboration on the strategy used above in~\eq{6} to~\eq{8}.
As an example of this paradigm (see~\cite{WGFC13a,WGFC13b}) consider the following computational problem: assume that you are provided with an experimental quantum system that you can evolve for a specified evolution time $t$.  We denote the transformation that this system enacts by $\mathcal{E}(t):\mathbb{C}^N \mapsto\mathbb{C}^N$ that can be applied to an arbitrary input quantum state for any $t\in \mathbb{R}$.  The computational problem is to estimate the true model for $\mathcal{E}(t)$, denoted $H$, in a (potentially continuous) family of models $\{\mathcal{H}\}$ using a minimal number of experiments and only efficient quantum computations.  

In order to practically learn $H$, we need to have a concise representation of the model.  We assume that each $H$ is parameterized by a vector $\vec{x}$ and explicitly denote a particular model $H$ as $H(\vec{x})$ when necessary.
As a clarifying example, let us return to the case of the Schr\" odinger equation~\eq{schrodinger} with the case where $V(r)=\frac{1}{2}k r^2$:
\begin{equation}
i\hbar \frac{\partial \psi(r,t)}{\partial t} = \left[\frac{-\hbar^2}{2m}\nabla^2 +\frac{1}{2}k r^2 \right] \psi(r,t)\label{eq:schrodinger2}.
\end{equation}
This problem, which corresponds to a quantum mechanical mass--spring system (harmonic oscillator), it may be that neither the mass of the particle $m$ nor the spring constant $k$ are known accurately.  In this case, the model parameters can be thought of as a vector $\vec{x} = [m,k]$ and the model itself is specified by $H(\vec{x}) = \left[\frac{-\hbar^2}{2m}\nabla^2 +\frac{1}{2}k r^2 \right]$.  Such models are known in quantum mechanics as Hamiltonians and they uniquely specify the quantum dynamical system: $\psi(r,t) = e^{-iH(\vec{x}) t/\hbar} \psi(r,0)$.  For finite--dimensional systems the rule is exactly the same: $\vec{v} \mapsto e^{-iH(\vec{x})t} \vec{v}$ (in units where $\hbar=1$), which means $\mathcal{E}(t)=e^{-iHt}$.

The problem of inferring the model parameters, $\vec{x}$, can be addressed by using Bayesian inference by following steps identical to those discussed in~\eq{6} to~\eq{8} in the context of Boson sampling.  The following theorem illustrates that the  likelihood functions that can be efficiently computed using a quantum computer for many quantum systems.

\begin{theorem}\label{thm:likelihood}
Let $\mathcal{E}(t)=e^{-i Ht}$ where $H\in \mathcal{H}$ is the model ``Hamiltonian'' and we assume that
\begin{enumerate}
\item For every $X\in\mathcal{H}$, $X\in \mathbb{C}^{N\times N}$ is a $d$--sparse Hermitian matrix whose non--zero entries  in each row can be efficiently located and computed.
\item There exists $\Lambda\in\mathbb{R}$ such that every $X\in\mathcal{H}$, $|X_{i,j}|\le \Lambda$ for all $i,j$.
\item $t$ is an arbitrary real valued evolution time for the quantum system that is specified by the user.
\end{enumerate}
then for any such hypothetical model, the likelihood of obtaining a measurement outcome $D$, $\Pr(D|H)$, under the action of $\left(e^{-iH_-t}\right)^{-1} \mathcal{E}(t)=e^{iH_-t}\mathcal{E}(t)$, where $H_-\in \mathcal{H}$, can be computed with high probability within error $\epsilon$ using a number of accesses to the entries of the matrices $H$ amd $H_-$ that, at most, scales as
$$
\frac{(\log^{*}(N))^{1+o(1)}d^{3+o(1)}(\Lambda t)^{1+o(1)}}{\epsilon^{1+o(1)}},
$$
and a number of elementary operations that scales polynomially in $\log(N)$ given that the initial state can be efficiently prepared using a quantum computer.
\end{theorem}
\begin{proof}
The proof follows by combining two well known quantum algorithms.  First the Childs and Kothari simulation algorithm~\cite{CK11} shows that the action of $e^{-iHt}$ can be simulated within error at most $\epsilon/2$ using 
$$\frac{(\log^*(N))^{1+o(1)} d^{3+o(1)}(\Lambda t)^{1+o(1)}}{\epsilon^{o(1)}},$$ 
accesses to the matrix elements of $H$.  The likelihood $\Pr(D|H)$ can then be estimated to within precision $\epsilon/2$ by repeating this algorithm $O(1/\epsilon^2)$ times and measuring the fraction of times that outcome $D$ is observed.  The overall error is at most $\epsilon$ since these errors are at worst additive~\cite{NC10}.  

A faster quantum algorithm called amplitude estimation exists for such sampling problems~\cite{BHM+00}.  It can be used to estimate the likelihood using a number of queries that scales quadratically better with $\epsilon$ than statistical sampling and is successful at least $81\%$ of the time.  The caveat is that we need to be able to reflect quantum state vectors about the space orthogonal to the marked state and also the initial state, which is efficient under the assumptions made above.  Both algorithms require a number of auxiliary operations that scale polynomially in $\log(N)$ which justifies the above claims of efficiency.
\end{proof}

This shows that the likelihood function can be efficiently approximated within constant error  for a broad class of quantum systems using a quantum computer: quantum chemistry, condensed matter systems and many other systems fall into this class. Now let us assume that $M$ possible hypothetical models are posited to describe $\mathcal{E}(t)$ that each satisfy the properties of~\thm{likelihood}.  Using~\eq{bayes} requires $M$ simulations resulting in an overall cost of
$$
\frac{M(\log^{*}(N))^{1+o(1)}d^{3+o(1)}(\Lambda t)^{1+o(1)}}{\epsilon^{1+o(1)}}.
$$
Here $N$ can be as large as $2^{200}$ for reasonably large quantum systems and $d$ is often on the order of a few hundred.  The cost will therefore likely be modest if $M$ is small, the Hamiltonian matrix is sparse and ($\Lambda t$) is not unreasonably long.  These requirements can be met  in most physically realistic cases.  

These results show efficiency for fixed $\epsilon$, but does $\epsilon$ need to be to prohibitively small to guarantee stability?  If, for example, $D$ is an outcome such that $\Pr(D|H)\in O({\rm poly}(N^{-1}))$ for hypothetical model $H$ then $\epsilon$ will have to be extremely small in order to compute the likelihoods to even a single digit of accuracy.  Cases where every model predicts $P(D|H(\vec{x}))\in \Theta(1/N)$ are therefore anathema to Bayesian inference.  

Fortunately, inversion removes this possibility because it reduces each experiment to two effective outcomes.  If $H_-=H$ then $e^{iH_- t}\mathcal{E}(t) \vec{v}  = \vec{v}$ for all $\vec{v}$, which precludes the possibility of exponentially small likelihoods if $\vec{v}$ is a computational basis vector (or more generally if $\vec{v}$ can be efficiently transformed to a computational basis vector using the quantum computer).  Conversely, it is well known that with high probability over models, $e^{iH_- t}\mathcal{E}(t) \vec{v} \ne \vec{v}$ in the limit as $N\rightarrow \infty$ if $H\ne H_-$ and $t>0$ is taken to be a constant~\cite{Haa10}.  In contrast if $t$ is small then it is trivial to see that almost all models and choices of $H_-$ will result in $e^{iH_-t}\mathcal{E}(t) \vec{v} = \vec{v}$.  Thus:
\begin{enumerate}
\item Each experiment has two outcomes: either $e^{iH_-t}\mathcal{E}(t): \vec{v} \mapsto \vec{v}$ or $e^{iH_-t}\mathcal{E}(t): \vec{v} \mapsto \vec{v}^{\perp}$ where $\vec{v}^{\perp}$ is in the orthogonal complement of $\vec{v}$.
\item The variable $t$ can be chosen by the user and hence can be chosen to ensure that roughly half of the most likely models for the quantum system will approximately yield $\vec{v}$ with high probability over models~\cite{WGFC13a}. 
\end{enumerate}
These properties will typically allow Bayesian inference to identify the correct model using a logarithmic number of experiments, similar to binary search. 
\alg{qhl} provides a concrete method that uses these principles to learn an approximate model for a quantum system.   

\begin{algorithm}[t!]
\caption{\label{alg:qhl} Quantum Hamiltonian learning algorithm.}
\rule{\linewidth}{1pt}
\begin{algorithmic}
\Require Prior probabilities $w_i$, $i \in \{1, \dots, M\}$, hypothetical model specifications $\vec{x}_i$,  $i \in \{1, \dots, M\}$, number of samples used to estimate probabilities $N_{\rm samp}$, total number of updates used $N_{\rm exp}$, state preparation protocol $Q$ for the initial states $\vec{v_0}$, protocol for implementing measurement operator $P$ such that $\vec{v_0}\vec{v_0}^*$ is a POVM element, an error tolerance $\epsilon$ and $k$ the number of votes used to boost the success probability of amplitude estimation.
\Ensure  Hamiltonian parameters $\vec{x}$ such that $H(\vec{x})\approx H(\vec{x}_{\rm true})$.

\vskip0.2em
\hrule
\vskip0.2em

\Function{QHL}{$\{w_i\}$, $\{\vec{x}_i\}$, $N_{\rm samp}$, $N_{\rm exp}$, $P$, $Q$,$\epsilon,k$}
  \For{$i \in 1 \to N_{\exp}$}
	\State $\vec{v_0} \gets Q(i)$. \inlinecomment{Prepare initial state.}
	\State Draw $\vec{x}_-$ and $\vec{x}'$ from $\Pr(\vec x):= w_i/\sum_i w_i $.
	\State $t\gets 1/\|H(x)-H(x')\|$.\inlinecomment{Choose $t$ according to guess heuristic}
	\State $H_-\gets H(\vec{x}_-)$.
	\State $D\gets $ measurement of $e^{iH(\vec{x_-})t}\mathcal{E}(t)\vec{v_0}$ using $P$.\inlinecomment{Perform experiment on untrusted system.}
	\For{$j \in 1 \to M$}\inlinecomment{Compute likelihoods using quantum computer}
	\For{$\kappa \in 1 \to 2k-1$}
		\State  $Y_\kappa \gets$ result of amplitude estimation  \inlinecomment{Estimate $\Pr(D|H(\vec{x_j}))$}
		\State \qquad\qquad on $e^{iH_-t}e^{-iH(x_j) t}\vec{v}_0$, to within error $\epsilon$.
	\EndFor
	\State $p_m \gets {\rm median}(Y)$.	\inlinecomment{Vote on correct likelihood}
	\EndFor
  \State $Z\gets \sum_{m=1,M}w_m p_m$.
  \State $w_i \gets w_i p_i/Z$.\inlinecomment{Perform  update.}
  \EndFor
  \State \Return $\sum_m w_m \vec{x}_m$\inlinecomment{Return Bayes estimate of $\vec{x}_{\rm true}$.}
\EndFunction
\end{algorithmic}
\rule{\linewidth}{1pt}
\end{algorithm}

The question is, how well does \alg{qhl} work in practice?  Since we lack quantum hardware that can implement it directly, we cannot assess  \alg{qhl}'s performance directly; however,  we can estimate how it ought to scale in practice using small numerical experiments.  We will see that it works extremely well for the examples considered.

A good benchmark of the performance of the algorithm is given in~\cite{WGFC13a}, wherein it is shown that this approach can be used to learn inter--atom couplings of frustrated Ising models (which are used condensed matter physics to model magnetic systems).  
The Hamiltonian matrix that describes these systems is
$$
H(\vec{x})= \sum_{i=0}^{n-1} \sum_{j=i+1}^n \vec{x}_{i,j} \begin{bmatrix} 1 & 0 \\ 0 & -1 \end{bmatrix}^{(i)} \otimes \begin{bmatrix} 1 & 0 \\ 0 & -1 \end{bmatrix}^{(j)},
$$
where this notation means that the state is assumed to be a tensor product of $n$ two--dimensional subsystems and each term in $H(\vec{x})$ acts non--trivially only on the $i^{\rm th}$ and $j^{\rm th}$ subsystems.  
This application is slightly more complex than the simple one discussed in~\alg{qhl} because the model is continuous rather than discrete, but by discretising the problem and using a technique called resampling, this process essentially reduces to that in~\alg{qhl}.

\fig{line_nonoise_plots} shows that the number of updates needed to learn the couplings in an Ising model that has $n=3,5,7,9$ qubits with pairwise interactions between every qubit in the system (i.e. the interaction graph is complete) where $\vec{v_0} = [1/\sqrt{2^n}, \ldots,1/\sqrt{2^n}]^t$.  The error decreases exponentially with the number of updates for any fixed $n$.  In other words, there exist $\gamma(n)>0$ such that the error scales as $e^{-\gamma(n) N_{\rm samp}}$. \fig{exp_scale} shows that the decay constants, $\gamma(n)$, for these exponential curves scale as $\Theta(1/{\rm dim}(\vec{x}))$, where ${\rm dim}(\vec{x})$ is the number of parameters in the model.  In general ${\rm dim}(\vec{x})$ need not be a function of $n$ but for the data above, each qubit interacts with all other qubits in the system and hence ${\rm dim}(\vec{x})=\binom{n}{2}=O(n^2)$.  We find that $\vec{x}$ can be learned to within error $\Delta$ using a number of experiments that empirically scales as $\Theta(n^2\log(1/\Delta))$.

The prior results allow us to estimate the complexity of inferring an Ising model $H\in \mathbb{C}^{2^n \times 2^n}$ using a quantum computer as an experimental device.  The algorithm requires a number of queries to $\mathcal{E}(t)$ that scales as $O(n^2\log(1/\Delta))$ and a number of elementary operations that scale as
$$
\frac{Mn^{2+o(1)}}{(\Delta\epsilon)^{1+o(1)}},
$$
where $M$ is the number of models used in the discretization ($M$ is known to scale sub--exponentially with $n$ but experimentally $M\in O({\rm polylog}(n))$ and a constant value of $\epsilon=0.05$ seems to suffice) and $d=1$.

\begin{figure*}[t!]
\begin{minipage}{0.45\linewidth}
\includegraphics[width=\textwidth]{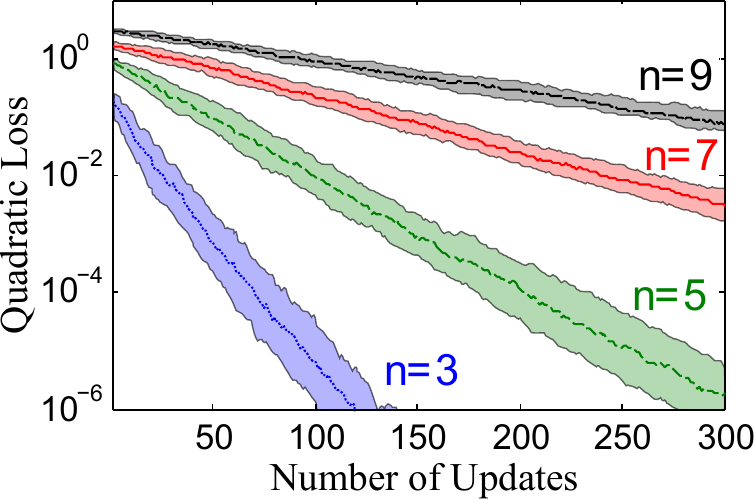}
\caption{Quadratic loss vs the number of experiments needed to learn Ising models.  The shaded regions are a $50\%$ confidence interval and the solid line gives the median.  Figure from~\cite{WGFC13a}.}\label{fig:line_nonoise_plots}
\end{minipage}
\hspace{0.5cm}
\begin{minipage}{0.45\linewidth}
\vspace{-1cm}
\includegraphics[width=\textwidth]{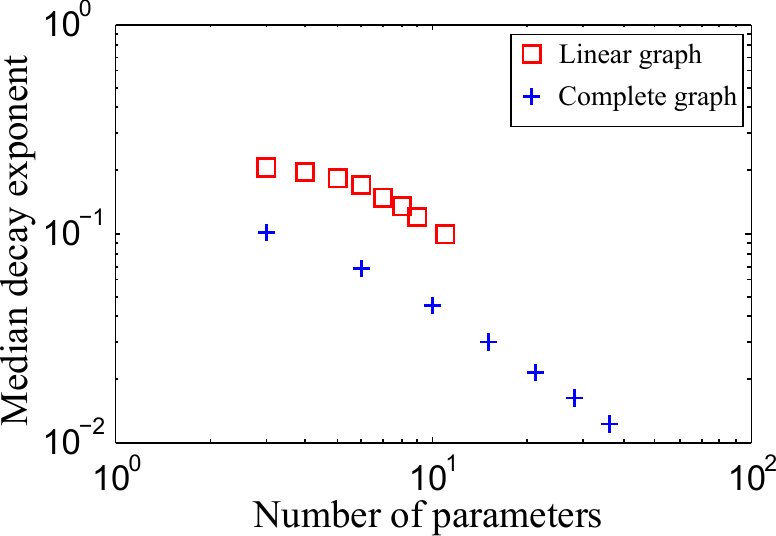}
\caption{The median decay exponent ($\gamma$) for the quadratic loss as a function of the number of parameters in the Ising model on the complete graph as well as the line.  Figure originally from~\cite{WGFC13a}.}\label{fig:exp_scale}
\end{minipage}
\end{figure*}

These scalings are not only show that a formal complexity can be assigned to the experimental problem of learning a physical model for a large quantum system, but also suggest that quantum computing could make the learning problem  tractable even for cases where $n>100$.  The combination of quantum computation and statistical inference may therefore provide our best hope of deeply understanding the physics of massive quantum systems that seem to be too complex to understand directly.

\section{Conclusion}
Returning now to our original question, we see that there is good reason (based on complexity theoretic conjectures) to suspect that solving the equations of quantum dynamics, such as the Schr\" odinger equation, may be intractable even for certain modestly large quantum systems.  This would seem to suggest that some quantum systems  exhibit dynamics that, for all practical purposes, cannot be directly compared to the predictions of quantum theory.   Quantum computation offers the possibility that such problems can be circumvented by indirectly comparing the dynamics of the system with simulations performed on the quantum computer.  This also provides a surprising insight: experimental quantum physics can be recast in the language of quantum computer science, allowing formal time complexities to be assigned to learning facts about nature and reinforcing the importance of computer science to the foundations of science. 

A quantum computer does not, however, allow us to address all of the deep problems facing us in physics:  they are not known to be capable of simulating the standard model in quantum field theory~\cite{JLP12}, which is arguably the most successful physical theory ever tested.  Simulation algorithms are also unknown for string theory or quantum loop gravity, meaning that many of the most important physics questions of this generation remain outside of reach of quantum computing.  Much more work may be needed before physicists and computer scientists can truly claim that quantum computers (or generalizations thereof) can be used to rapidly infer models for all physical systems.

\acknowledgements{
I would like to thank Christopher Granade and Yuri Gurevitch for valuable feedback and Robert Browne for the inspiration behind this work.
}

\bibliographystyle{unsrt}

\end{document}